\documentclass[12pt, onecolumn]{article}
\usepackage{mathpazo}
\usepackage{times}

\usepackage[utf8]{inputenc} 
\usepackage[T1]{fontenc}
\usepackage{url}
\usepackage{ifthen}
\usepackage{cite}
\usepackage{graphicx}
\usepackage{amssymb}
\usepackage[cmex10]{amsmath} 
\usepackage{amsthm}
\usepackage{color}
\usepackage[ruled,vlined,linesnumbered]{algorithm2e}
\usepackage{comment}
\usepackage{caption}
\usepackage{blkarray} 
\usepackage{siunitx}

\usepackage[margin=1in]{geometry}

\newcommand\blfootnote[1]{%
\begingroup
  \renewcommand\thefootnote{}\footnote{#1}%
  \addtocounter{footnote}{-1}%
\endgroup
}

\renewcommand{\le}{\leqslant}
\renewcommand{\leq}{\leqslant}

\newcommand{\be}[1]{\begin{equation}\label{#1}}
\newcommand{\ee}{\end{equation}}

\usepackage{amsthm}

\newtheorem{theorem}{Theorem}

\newcommand{\Cref}[1]{Co\-ro\-lla\-ry\,\ref{#1}}

\DeclareMathAlphabet{\mathbfsl}{OT1}{ppl}{b}{it} 

\begin{document}
\title{$\,$\\[-1.62ex] \textbf{
Partially Polarized Polar Codes:\\
A New Design for 6G Control Channels
\blfootnote{This paper has also been submitted to the \textit{IEEE International Conference on Communications (ICC) 2026}.}
\\[0.27ex]}}
\author{
  \begin{minipage}[t]{0.45\textwidth}
    \centering
    \textbf{Arman Fazeli}\\
    {\small Apple Inc.\\
    Cupertino, CA 95014, USA\\
    \texttt{afazeli@apple.com}}
  \end{minipage}\hfill
  \begin{minipage}[t]{0.45\textwidth}
    \centering
    \textbf{Mohammad M. Mansour}\\
    {\small Apple Inc.\\
    Cupertino, CA 95014, USA\\
    \texttt{m\_mansour@apple.com}}
  \end{minipage}\\[4.5em]
  \begin{minipage}[t]{0.45\textwidth}
    \centering
    \textbf{Ziyuan Zhu}\\
    {\small University of California San Diego\\
    La Jolla, CA 92093, USA\\
    \texttt{ziz050@ucsd.edu}}
  \end{minipage}\hfill
  \begin{minipage}[t]{0.45\textwidth}
    \centering
    \textbf{Louay Jalloul}\\
    {\small Apple Inc.\\
    Cupertino, CA 95014, USA\\
    \texttt{ljalloul@apple.com}}
  \end{minipage}
}

\maketitle
\begin{abstract}
\noindent
We introduce a new family of polar-like codes, called Partially Polarized Polar (PPP) codes. PPP codes are constructed from conventional polar codes by selectively pruning polarization kernels, thereby modifying the synthesized bit-channel capacities to ensure a guaranteed number of non-frozen bits available early in decoding. These early-access information bits enable more effective early termination, which is particularly valuable for blind decoding in downlink control channels, where user equipment (UE) must process multiple candidates, many of which carry no valid control information. Our results show that PPP codes offer substantial performance gains over conventional polar codes, particularly at larger block lengths where hardware limitations restrict straightforward scaling. Compared with existing methods such as aggregation or segmentation, PPP codes achieve higher efficiency without the need for additional hardware support. Finally, we propose several frozen-bitmap design strategies tailored to PPP codes.
\end{abstract}


\clearpage
\section{Introduction}
\label{sec:intro}
\noindent 
Polar codes, introduced by Arıkan in 2008~\cite{arikan2009channel}, were the first explicit family of error correction codes proven to achieve the capacity of symmetric binary input memoryless channels with efficient encoding and decoding algorithms. Since their invention, polar codes have attracted growing interest from both academia and industry, leading to the development of advanced decoding techniques such as successive cancellation list (SCL) decoding~\cite{tal2015list}, as well as numerous efficient hardware implementations~\cite{leroux2011hardware, leroux2013semiparallel, balatsoukas2014hardware,giard2015unrolled}. These advances have culminated in the adoption of polar codes in the 5G cellular standard~\cite{3gpp38212R15}, particularly for short to moderate block lengths.

In 5G New Radio (NR), downlink control information (DCI) is transmitted over the physical downlink control channel (PDCCH) and encoded using polar codes~\cite{3gpp38212R15}. Since the user equipment (UE) does not know in advance which control message is intended for it, blind decoding must be performed across multiple PDCCH candidates. Each candidate must be fully decoded and validated through a cyclic redundancy check (CRC) masked with the UE’s radio network temporary identifier (RNTI). This procedure is computationally demanding, as the decoder is required to process a large number of invalid candidates using full-length polar decoding, resulting in increased power consumption, latency, and hardware utilization~\cite{condo2017blind,jalali2020jdd,giard2017blind}. The challenge is particularly severe in low SNR conditions or when higher aggregation levels are used,re the number of blind decoding attempts increases substantially.

Looking ahead to sixth-generation (6G) systems, several studies and standardization trend reports indicate that the increasing complexity and diversity of control signaling (e.g., cross-cell or cross-carrier scheduling and richer resource configurations) will place additional pressure on DCI formats, likely requiring larger payload sizes~\cite{3gppWP}. This trend will further burden decoder hardware, especially under stringent low-latency and energy-efficiency constraints at the UE.

\begin{figure}[t!]
    \centering
    \includegraphics[width=\columnwidth,trim={0.5cm 7.0cm 13.5cm 1.0cm},clip]{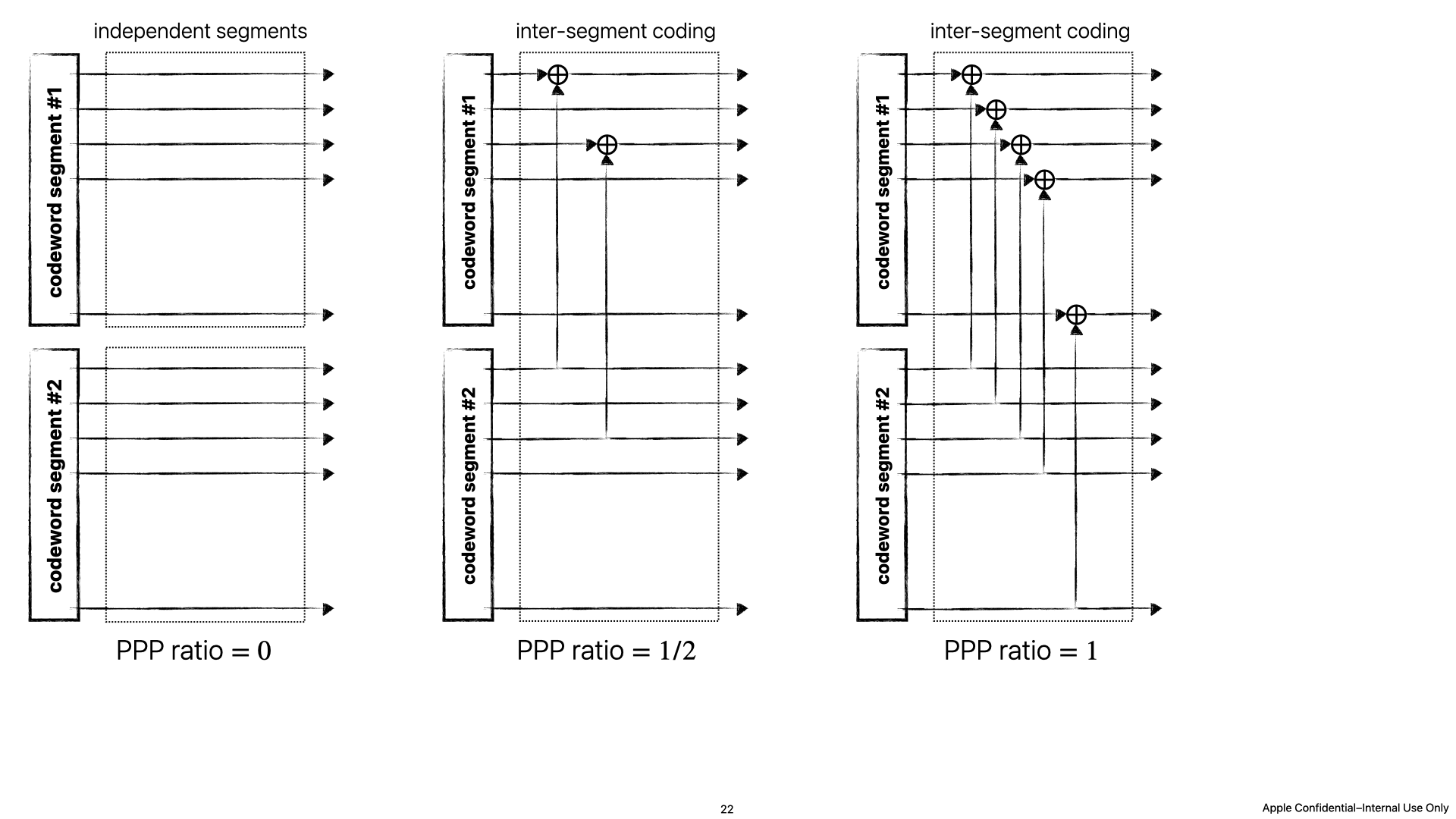}
    \caption{Illustration of inter-segment dependencies in PPP codes. The PPP ratio specifies the fraction of XOR operations retained in the final polarization layer: a ratio of $0$ removes all inter-segment dependencies, while a ratio of $1$ corresponds to a full polar transform, yielding maximum capacity imbalance.}\vspace{-2ex}
    \label{fig:intersegment}
\end{figure}

Assuming channel quality and noise levels remain constant, increasing the DCI payload requires a longer code to maintain the same overall code rate and error protection. However, commercial UE polar decoders already operate near practical hardware limits. This constraint is evident even in 5G NR, where more resource elements (REs) may be available than are used by the polar code, yet the code length is typically capped at $N=512$. Rather than extending code length—which would necessitate larger and more power-hungry hardware—the standard adopts repetition-based aggregation to enhance effective SNR under poor channel conditions. While aggregation improves reliability, it does not provide the coding gains that longer polar codes can deliver.

To prepare for longer DCI payloads, an alternative to increasing code length is to partition the payload into multiple segments, with each segment encoded, transmitted, and decoded separately. This method, known as codeword or DCI segmentation, offers hardware advantages: the same decoder can be reused across segments, potentially in parallel, thereby controlling additional decoding latency. However, similar to aggregation, segmentation sacrifices the coding gains of longer polar codes. Treating each segment independently prevents the system from exploiting polarization across a larger block length, reducing error-correction performance.

Even the loss of a single polarization layer can incur an SNR degradation on the order of $0.3$--$0.5\ \text{dB}$. This penalty is significant, particularly given that such an SNR improvement often requires doubling the list size in an SCL decoder. To compensate for the degradation introduced by segmentation, the UE would need to increase its list size accordingly, which nearly doubles both decoder area and power consumption. Thus, the supposed hardware savings of segmentation are effectively offset, undermining its appeal as a low-complexity solution.

\vspace{0.5em}
\noindent {\bf {Our contributions}} are threefold:

\begin{enumerate}
    \item \textbf{Inter-segment coding via partial polarization.} 
   We introduce a new inter-segment coding design based on partial polarization.
This additional coding layer restores much of the coding gain that would otherwise be lost with independent segments.
    
    \item \textbf{Integration with two-stage decoding.} 
    The proposed design aligns naturally with two-stage decoding, where an initial subset 
    of bits is decoded first and used to decide whether full decoding should continue. 
    This is particularly beneficial in blind decoding scenarios, where the UE attempts many invalid 
    candidates. Early termination based on the first-stage results can substantially reduce 
    unnecessary decoding. To ensure predictable hardware savings, the leading segment must 
    carry sufficient information (e.g., the RNTI sequence). Partial polarization enables such 
    flexibility by rebalancing the synthesized bit-channel capacities so that the leading 
    segment maintains a minimum code rate.
    
    \item \textbf{Partially Polarized Polar (PPP) codes.} 
    We introduce PPP codes, where each segment is itself a polar code. This structure allows efficient hardware reuse across segments while preserving inter-segment coding gain. We further propose strategies for frozen-bit map design and demonstrate that PPP codes outperform traditional segmentation and aggregation schemes.
\end{enumerate}

\vspace{0.5em}
\noindent {\bf {Prior work:}} Several approaches related to inter-segment coding, blind detection, and polar-coded modulation have been studied. Partially information-coupled (PIC) polar codes couple adjacent blocks by sharing systematic bits, improving transport block error rate relative to uncoupled polar codes~\cite{wu2018picpolar}. A blind detection scheme for polar codes based on two-stage decoding with UE-ID embedding and early stopping was proposed in~\cite{condo2018blinddetection}. The complexity–accuracy trade-off in blind detection was further analyzed in~\cite{giard2018tradeoff}, which introduced belief-propagation (BP)-based metrics for candidate pruning. Blind frame synchronization assisted by polar codes was studied in~\cite{feng2023blindfs}, leveraging soft information from frozen bits to integrate synchronization with SCL decoding. Polar coding has also been extended to modulation: Mahdavifar \textit{et al.} proposed a compound polar-coding method for bit-interleaved coded modulation (BICM) that achieves multichannel capacity~\cite{mahdavifar2016bicm}, while Mughal \textit{et al.} introduced multilevel polar constructions for cooperative spatial modulation with gains over Rayleigh and Rician channels~\cite{mughal2019ccsm}. A general framework for polar-coded modulation was developed in~\cite{seidl2013pcm}, unifying multilevel and BICM approaches with optimized labeling for improved performance.

\section{Capacity Re-balancing and Two-Stage Decoding}
\label{sec:capacity}
\noindent
We realize inter-segment coding by inserting a polarization-like layer, analogous to the final stage of a length-$N$ polar transform spanning all segments. In a standard polar code of length $N$, the last polarization layer applies $N/2$ XOR operations, where $N$ is the combined length of all segments. In contrast, a partial inter-segment polarization layer applies only $\tau N/2$ XOR operations, with $0 \le \tau \le 1$ denoting the \emph{partial polarization ratio}. Figure~\ref{fig:twoStageDecoder} illustrates the construction for two segments. When $\tau=0$, no XOR operations are performed and the segments remain independent. When $\tau=1$, all $N/2$ operations are applied, recovering the full polar transform and its maximum capacity polarization. Intermediate values (e.g., $\tau=1/2$) introduce dependencies for only a fraction of the coded bits in each segment, thereby re-balancing capacity across segments.

\begin{figure}[t!]
    \centering
    \includegraphics[width=\columnwidth,trim={0.5cm 3.0cm 8.5cm 1.0cm},clip]{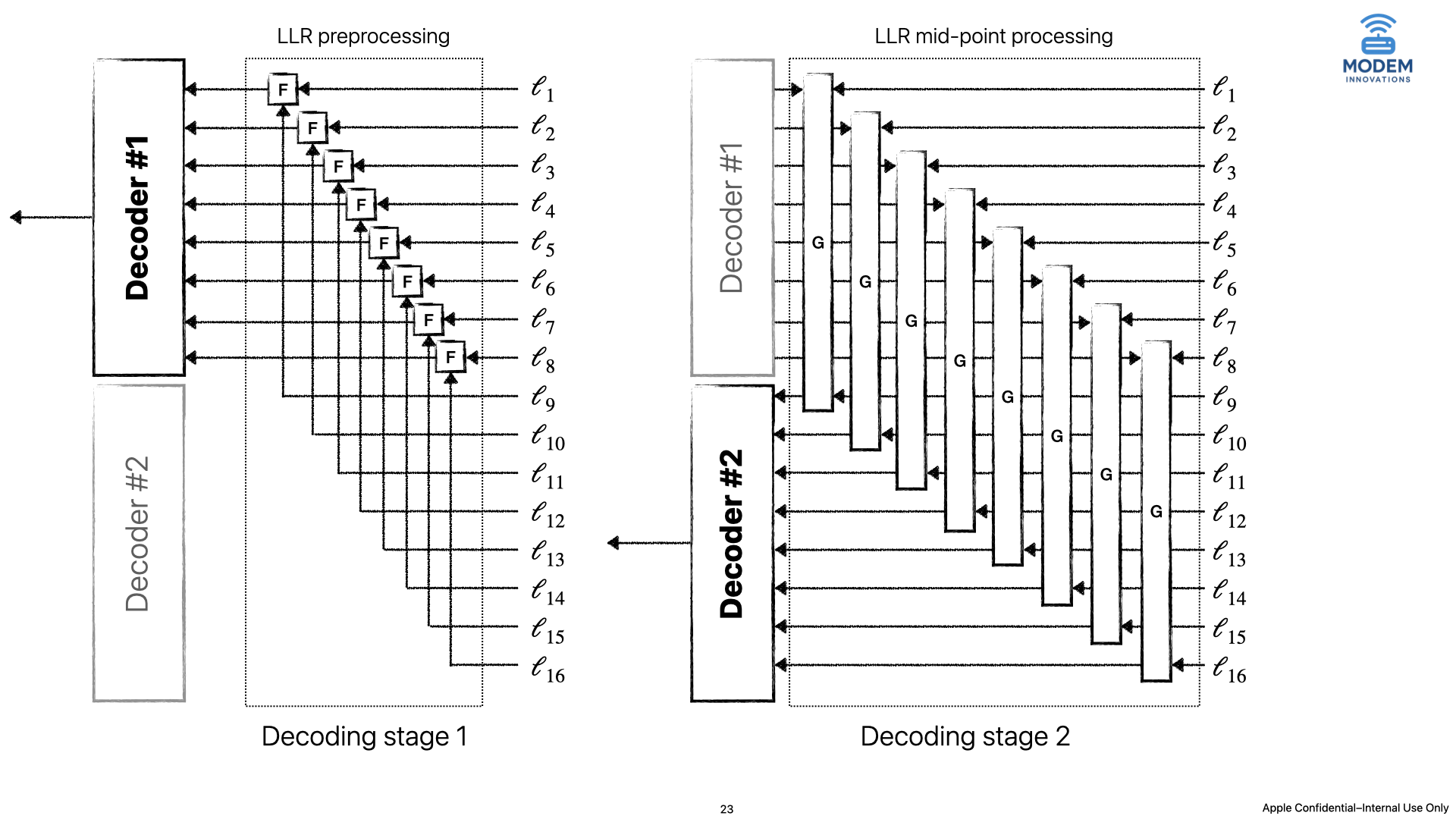}
    \caption{Two-stage decoding with partially polarized inter-segment coding. Channel LLRs are first pre-processed using $F$ operations to generate the input for the first decoder. After the first decoding stage, $G$ operations combine the results to construct the synthesized LLRs for the second decoder.}\vspace{-2ex}
    \label{fig:twoStageDecoder}
\end{figure}

Blind decoding at the UE is computationally demanding, as up to 44 PDCCH candidates may need to be decoded in 5G-NR, most of which do not carry valid DCI for the UE. Two-stage decoding reduces this complexity by enabling early termination after decoding only a fraction of the message. In 5G NR, the RNTI length is 16 bits. For two-stage decoding to provide meaningful early-termination opportunities, the first segment must carry a sufficient portion (often all) of the RNTI sequence, which imposes a minimum code-rate requirement on the leading segment.

Each XOR operation in the inter-segment coding layer corresponds to a polarization step 
that transforms two i.i.d. copies of a channel $W$ into a degraded synthesized channel $W^{-}$ and an improved synthesized channel $W^{+}$. Consider the two-segment design illustrated in Fig.~\ref{fig:twoStageDecoder}. With a full polarization layer (i.e., a conventional polar transform), the effective capacity allocated to the first segment is $\tfrac{N}{2} C(W^{-})$, which may fall below the required code rate. Partial polarization addresses this issue by increasing the effective capacity of the leading segment. Specifically, for a partial polarization ratio $\tau$, the effective capacity of the first and second segments respectively becomes
\begin{align}\label{eq:capacity}
\begin{split}
    C_{\text{first}} &= \tfrac{N}{2} \left( \tau C(W^{-}) + (1-\tau) C(W) \right), \\
    C_{\text{second}} &= \tfrac{N}{2} \left( \tau C(W^{+}) + (1-\tau) C(W) \right). 
\end{split}
\end{align}
Here, a $\tau$ fraction of pairs undergo polarization (contributing $C(W^{-})$ to the first and $C(W^{+})$ to the second segments), while a $(1-\tau)$ fraction remain unpolarized (contributing $C(W)$), thus re-balancing capacity toward a smaller gap.

As an example, consider $N = 128$ with an overall code rate corresponding to $24$ information bits and $16$ RNTI bits. Assume the underlying channel is a binary erasure channel (BEC) with erasure probability $z = 0.5$. Under a conventional polar transformation, the first segment receives $64$ synthesized $W^{-}$ channels, each with erasure probability
\begin{equation*}
    1 - (1-z)^2 = 0.75.
\end{equation*}
The resulting effective capacity of the leading segment is
\begin{equation*}
    C_{\text{first}} = 64 \times (1-0.75) = 16 \;\text{bits per channel use}.
\end{equation*}
This capacity is insufficient to reliably transmit $16$ RNTI bits, since the achievable 
communication rate must be strictly less than the channel capacity according to Shannon’s theorem. With a partial polarization layer of ratio $\tau = 0.5$, the effective capacity of the leading 
segment becomes
\begin{equation*}
    C_{\text{first}} = 64 \times \big( \tau (1-0.75) + (1-\tau)(1-0.5) \big) 
    =  24.
\end{equation*}
This value exceeds the design code rate, ensuring that the first segment can reliably 
accommodate the $16$ RNTI bits. 

\begin{figure}[t!]
    \centering
    \includegraphics[width=\columnwidth,trim={0.0cm 2.5cm 1.5cm 0.0cm},clip]{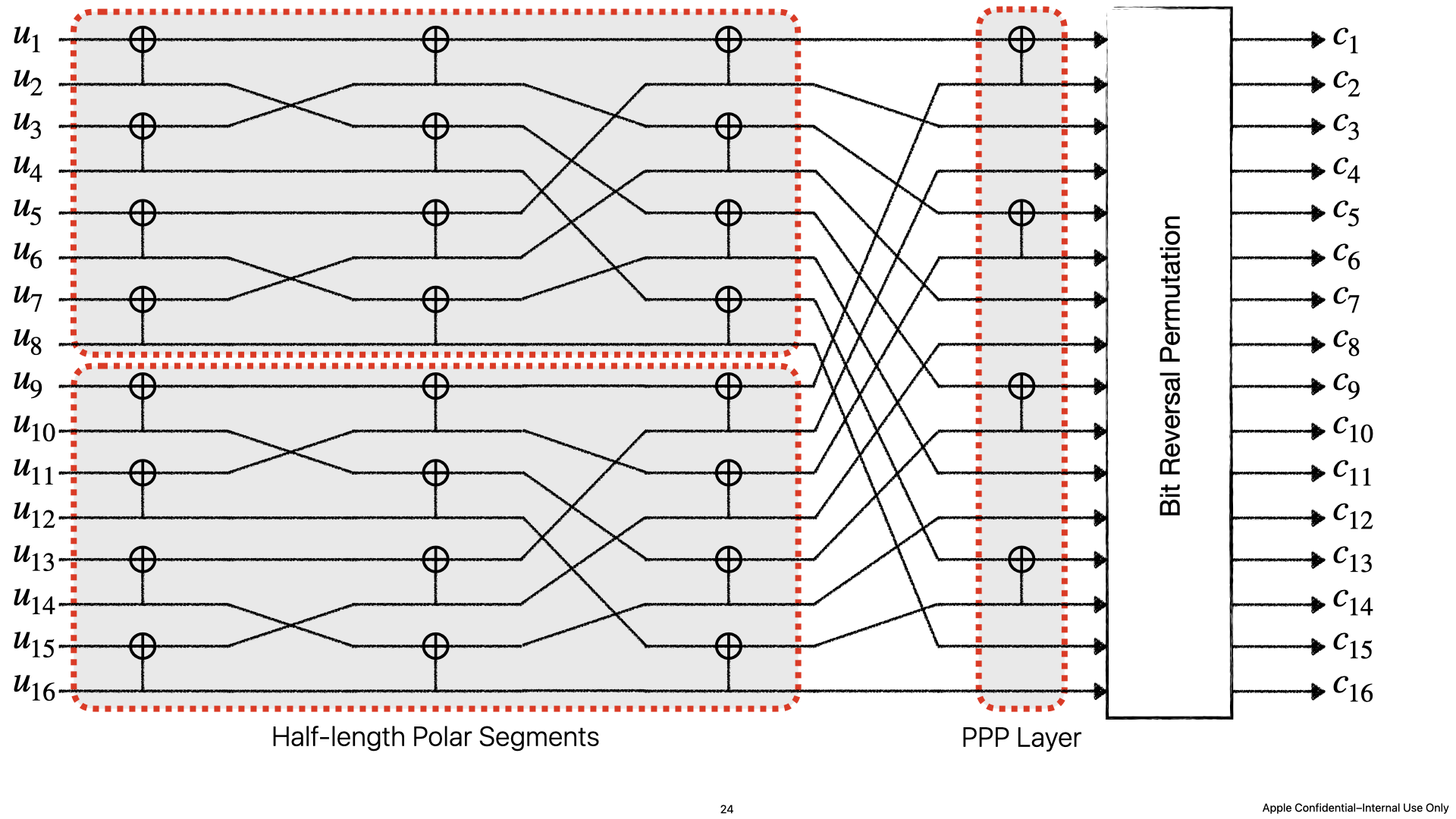}
    \caption{Example of a PPP code with partial ratio $\tau = 1/2$.}\vspace{-2ex}
    \label{fig:pppHW}
\end{figure}

Two-stage decoding with inter-segment partial polarization does not restrict the individual segment codes to be polar codes. In particular, if the leading segment is designed for an extremely low rate, it may be advantageous to replace it with a code of higher minimum distance, enabling fast ML-like decoding. Alternatively, the first stage can be simplified to a detection task: the UE only needs to determine whether its RNTI is the most likely estimate, and proceeds to the second decoding stage only if the answer is affirmative.

\section{Partially Polarized Polar Codes}
\label{sec:PPP}
\noindent 
In this section, we introduce Partially Polarized Polar (PPP) codes. A PPP code is formed by partitioning the payload into multiple segments, encoding each segment with a conventional polar code, and then applying a partial polarization layer across segments. This construction preserves the modularity of segmentation (enabling reuse of existing polar-decoder hardware) while introducing controlled inter-segment dependencies that recover much of the coding gain otherwise lost under independent segmentation.

The partial polarization ratio $\tau$ is a key design parameter for PPP codes. It is determined by factors such as block length, channel SNR, the minimum code rate required for the leading segment, and the target block error rate (BLER). Although the choice of $\tau$ may seem heuristic, it often suffices to select $\tau$ so that the effective capacity of the leading segment matches the design requirement. Concretely, one can simulate a length-$N/2$ polar code at the target code rate to identify the SNR corresponding to the desired BLER, and then choose $\tau$ to achieve an equivalent effective SNR by re-balancing capacities according to~\eqref{eq:capacity}.

\begin{figure}[t!]
    \centering
    \includegraphics[width=\columnwidth,trim={0.0cm 15.0cm 5.5cm 0.0cm},clip]{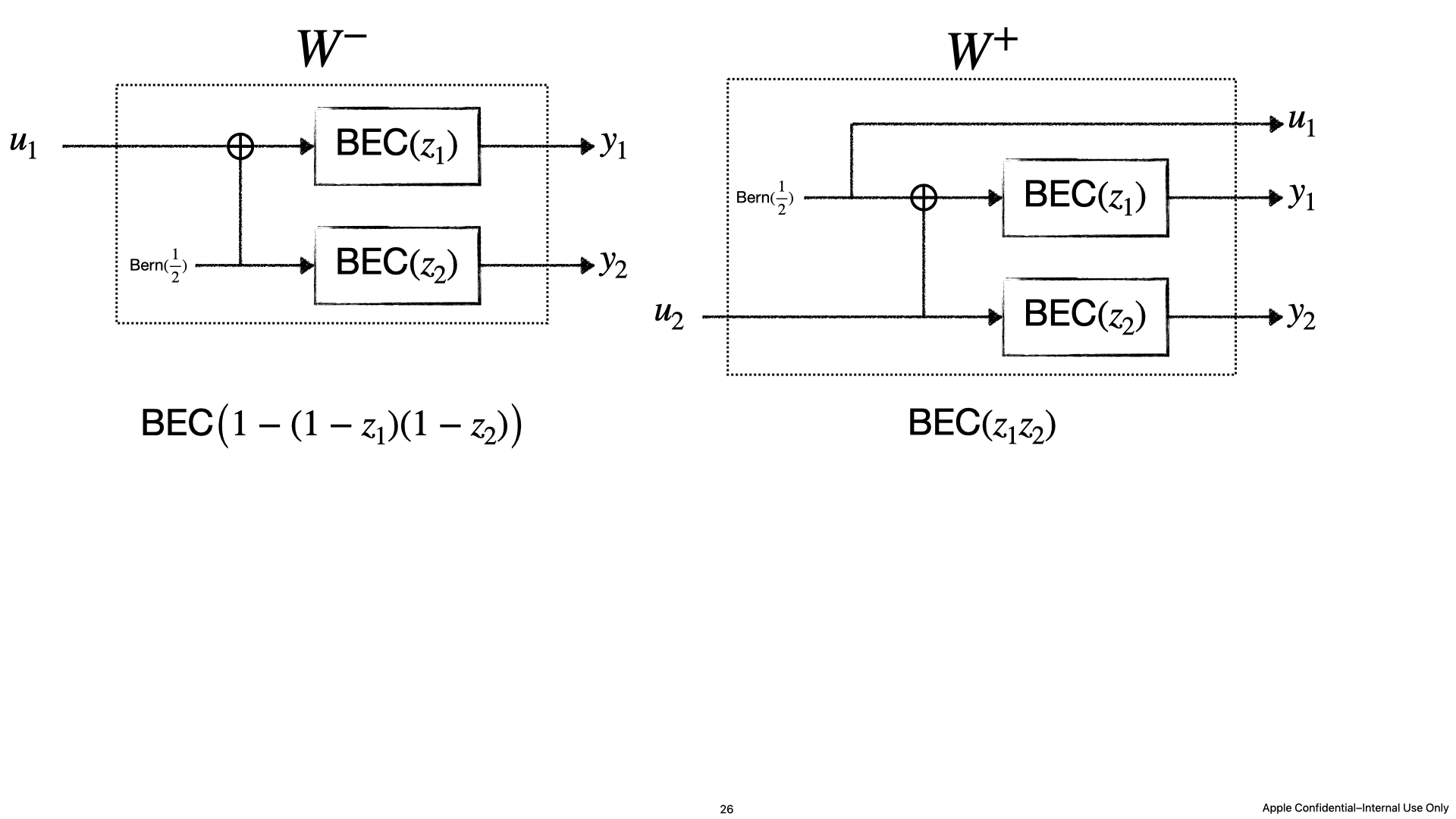}
    \caption{Synthesized erasure channels obtained by combining two non-identical erasure channels. The Bhattacharyya parameters of the erasure channels are equal to their erasure probabilities.}\vspace{-2ex}
    \label{fig:bec}
\end{figure}

The value of $\tau$ specifies the number of XOR operations retained in the partial polarization layer. A natural design question is which subset of the $N/2$ XOR operations from the full polarization layer should be kept to form the partial inter-segment coding. To address this, we track the polarization process and examine how partial polarization affects the evolution of the Bhattacharyya parameter in the simplified case of the binary erasure channel (BEC).

Recall that in the polarization process, combining two erasure channels produces two 
synthesized erasure channels. \linebreak Specifically, if the original channels have erasure probabilities $z_1$ and $z_2$, then the noisier synthesized channel, $W^-$, is
\begin{equation*}
    \text{BEC}\!\left( 1 - (1-z_1)(1-z_2) \right),    
\end{equation*}
while the less noisy synthesized channel, $W^+$, is
\begin{equation*}
    \text{BEC}(z_1 z_2).
\end{equation*}
For BECs, the corresponding Bhattacharyya parameters coincide with the erasure probabilities. These synthesized channels are illustrated in Fig.~\ref{fig:bec}.

The partial polarization layer effectively creates non-identical channels, which are then further combined through polarization within each segment. For example, in Fig.~\ref{fig:pppHW}, the synthesized input to the first segment consists of multiple copies of the underlying channel $W$ and multiple copies of $W^{-}$. These synthesized channels are subsequently combined through the remaining polarization stages. Polarizing non-identical channels is well understood: Mahdavifar \textit{et al.}~\cite{mahdavifar2013compound} proposed compound polar codes (as an alternative to segmentation) and used bit-interleavers to improve coding gains at higher modulation orders. Separately, Alsan and Telatar~\cite{alsan2016simple} showed that Arıkan's construction polarizes non-stationary memoryless channels and achieves channel capacity.

A key observation in~\cite{mahdavifar2013compound} is that, when combining non-identical channels, it is beneficial to reorder them so that, after a few polarization steps, the inputs to each kernel are identical (or nearly identical) in reliability. For $\tau = 1/2$ in Fig.~\ref{fig:pppHW}, this corresponds to alternating copies of $W$ and $W^{-}$, ensuring that after one polarization step the inputs to subsequent kernels are identical pairs. Numerical simulations indicate that similar reordering strategies are also effective for other partial polarization ratios. For instance, for $\tau = 1/4$, the partial polarization layer can be arranged so that, after the bit-reversal permutation, the synthesized input for the first segment appears as
\[
    W^{-},\ W,\ W,\ W,\ W^{-},\ W,\ W,\ W,\ \dots
\]
We leave the theoretical proof of the optimality of such constructions to future work. However, it is important to note that PPP codes with any fixed partial polarization ratio are also capacity-achieving in the same sense as conventional polar codes. This result is formalized in the theorem below.

\begin{theorem}\label{thm:capacity}
    For any fixed $\tau = \tfrac{\lambda}{\Lambda}$, where $\Lambda = 2^m$ for some 
    integer $m$ and $\lambda \in [0,\Lambda]$, PPP codes achieve the 
    symmetric channel capacity under successive cancellation (SC) decoding.
\end{theorem}

\begin{proof}
    The proof is conservative, since we disregard all polarization effects contributed by the 
first $m$ layers. A PPP code with $\tau = \lambda / 2^m$ is formed by retaining the first $\lambda$ XOR operations out of every $2^m$ in the inter-segment layer. After $m$ layers of polarization, up to $2^m$ different bit-channels are synthesized, and from that point onward, the inputs to every polarization kernel are identical.  

Let the underlying communication channel be denoted by $W$, and let the synthesized channels after $m$ stages be $$W_1, W_2, \dots, W_{2^m}$$ with corresponding symmetric capacities $$I(W_1), I(W_2), \dots, I(W_{2^m}).$$
Since the polarization process preserves capacity, we have
\begin{equation*}
        I(W) = \frac{1}{2^m} \sum_{j=1}^{2^m} I(W_j).
\end{equation*}
The PPP layer effectively decomposes the overall code of length $n$ into $2^m$ segments of conventional length-$N'$ polar codes, where the symmetric capacities of the synthesized bit-channels assigned to these segments are given by $I(W_j)$ for $j = 1, \dots, 2^m$ and $N' = N/2^m$. From this point onward, we can directly invoke the original polarization theorem of Arıkan~\cite{arikan2009channel}.  

Assume an overall code rate $R < I(W)$. Define the segment rates as
\begin{equation*}
    R_j = \frac{R}{I(W)} \cdot I(W_j), 
    \qquad j = 1, 2, \dots, 2^m.
\end{equation*}
By Arıkan’s theorem, each of these $2^m$ polar segments can be made to communicate reliably at rate $R_j$, with error probabilities that vanish as the block length increases. Let $\epsilon_{N'}$ denote the maximum block error probability under SC decoding among these length-$N'$ segments. Therefore, as $N'$ increases, we have $\epsilon_{N'} \rightarrow 0$. Consequently, the overall block error probability of the PPP code satisfies
\begin{equation*}
    P_e(N) \leq 2^m \, \epsilon_{N'},
\end{equation*}
which converges to zero as $N \to \infty$, since $m$ is fixed and independent of $N$.
\end{proof}

Although the proof establishes that PPP codes can achieve channel capacity by configuring the polar segments individually, a more effective construction is obtained by computing the reliability values of all synthesized bit-channels and selecting the best ones from the combined pool. This selection must also satisfy the rate constraints imposed by the two-stage decoding requirement. Note that PPP codes can be designed in full compliance with the 5G NR polar codes for each individual segment. For a given $(N,k)$ configuration, it suffices to provide additional instructions on how to partition the control payload into PPP segments.

Reliability values can be obtained through several methods: 1) approximating the channels with binary erasure channels (BECs) of equal capacity and computing the corresponding Bhattacharyya parameters~\cite{arikan2009channel}, 2) using the Tal--Vardy construction~\cite{tal2013construct}, 3) applying Gaussian approximation~\cite{trifonov2012efficient}, 4) performing density  evolution~\cite{mori2009performance}, or 5) conducting extensive Monte Carlo 
simulations~\cite{korada2009polar}. The only modification in PPP codes occurs in the first polarization layer, which is replaced by the partial polarization. Therefore, the 
construction algorithm must carefully track the specific channel combinations applied at each stage.

While these methods often yield highly optimized frozen bitmaps for successive cancellation (SC) decoding, they are less effective under successive cancellation list (SCL) decoding. This is because SCL decoding, particularly with large list sizes, approaches the performance of maximum likelihood (ML) decoding. In this case, the minimum distance of the code and its overall weight distribution play a more significant role than the bit-channel reliabilities computed under the assumption of a purely sequential decoder. To address this, the $\beta$-expansion method~\cite{he2017beta} has been proposed as a prominent approach for constructing reliability sequences tailored to SCL decoding.

For a given $j \in \{1,2,\dots,N=2^n\}$, let
    $(a_n,a_{n-1},\dots,a_1)_2$
denote the binary expansion of $j-1$. The $\beta$-expansion method assigns the following reliability metric to the $j$-th bit-channel in a conventional polar code:
\begin{equation*}
    m_j = \sum_{t=1}^n a_t \beta^t,
\end{equation*}
where $\beta \approx 2^{1/4}$. Bit-channels with larger metrics are considered more reliable. This aligns with prioritizing higher-weight rows in the polar transformation matrix, thereby improving the minimum distance properties of the overall code.

Note that the leading coefficient $a_n$ determines whether the bit-channel belongs to the first or the second segment. In a conventional polar code, the first polarization step has the greatest impact on the final reliability, and thus it is assigned the largest multiplier $\beta^n$. In PPP codes, however, this might not hold. For example, when $\tau = 0$, there is no inter-segment coding, and hence it is irrelevant which segment the bit-channel belongs to. Conversely, when $\tau = 1$, the PPP code reduces to a conventional polar code. Motivated by this observation, we propose the following modified metric, referred to as the 
$\alpha\beta$-expansion:
\[
    m_j = \alpha(\tau)a_n\beta^n + \sum_{t=1}^{n-1} a_t \beta^t,
\]
where $\alpha(\tau) \in [0,1]$ is a parameter dependent on the partial polarization ratio, satisfying
\[
    \alpha(\tau) = 0 \quad \text{for } \tau = 0, 
    \qquad
    \alpha(\tau) = 1 \quad \text{for } \tau = 1.
\]
A straightforward but non-optimal choice is $\alpha(\tau) = \tau$. More refined 
configurations may be obtained by exploiting partial order properties. We leave further 
optimization of this construction method to future work.

Recently, reinforcement learning--based sequence construction methods, particularly those using graph neural networks (GNNs)~\cite{liao2023scalable}, have been shown to produce highly effective reliability sequences. These approaches are more adaptable to realistic scenarios, such as when the underlying channel deviates from the simple AWGN model. Moreover, they can be trained to optimize code performance for specific code rates of practical interest, which is especially relevant to control channel standards, where polar codes are typically deployed at low code rates.  

We also developed a GNN-based reinforcement learning framework to generate reliability sequences for PPP codes. In addition to replacing the conventional polar generator matrix with the PPP graph, the reinforcement learning policy must be adapted to account for two-stage decoding. For instance, if an 8-bit CRC is placed on the upper half and a 16-bit CRC on the lower half, then effective CRC utilization requires preserving at least 16 unfrozen positions in the upper half and at least 32 in the lower half. Consequently, the learning agent halts freezing in the upper half once only 16 unfrozen positions remain, and similarly stops in the lower half when only 32 remain.

Figure~\ref{fig:pppCon} compares the performance of different PPP constructions for code length $N=64$ and payload size $K=32$. The first segment is protected with a 4-bit CRC, while the second segment uses an 8-bit CRC. The list size is set to $L=16$, which provides a performance close to maximum likelihood decoding. As shown, all three construction methods, namely Bhattacharyya tracking, $\alpha\beta$-expansion, and GNN-based construction, achieve desirable performance. In particular, PPP codes significantly outperform full segmentation and nearly close the gap to the larger polar code. The 5G polar curves are obtained by simulating the polar codes constructed based on the 5G NR polar sequence in~\cite{3gpp38212R15}.

For the GNN-based construction, the training configuration was as follows: the learning rate was $10^{-5}$, training was performed for \num{7000} episodes with a batch size of \num{32}, and the training SNR for each episode was uniformly sampled from the range $[2,3]$ dB. The trained sequence is presented in Table~\ref{table:n64seq}.  
\begin{table}[h]
\centering
\small
\begin{tabular}{ccccccccc}
1 & 2 & 3 & 5 & 9 & 17 & 4 & 33  & \dots \\
13 & 21 & 37 & 41 & 14 & 19 & 35 & 12 & \dots\\\
34 & 15 & 22 & 10 & 49 & 11 & 39 & 18 & \dots\\\
38 & 7 & 42 & 25 & 43 & 20 & 50 & 51 & 
\end{tabular}
\caption{The 32 least reliable bit-channels (in sorted order) for $N=64$, $L=16$, and $\tau=1/2$, as computed by the GNN.}
\label{table:n64seq}
\end{table}

\begin{figure}[t!]
    \centering
    \includegraphics[width=\columnwidth,trim={0.0cm 6.5cm 25.5cm 0.0cm},clip]{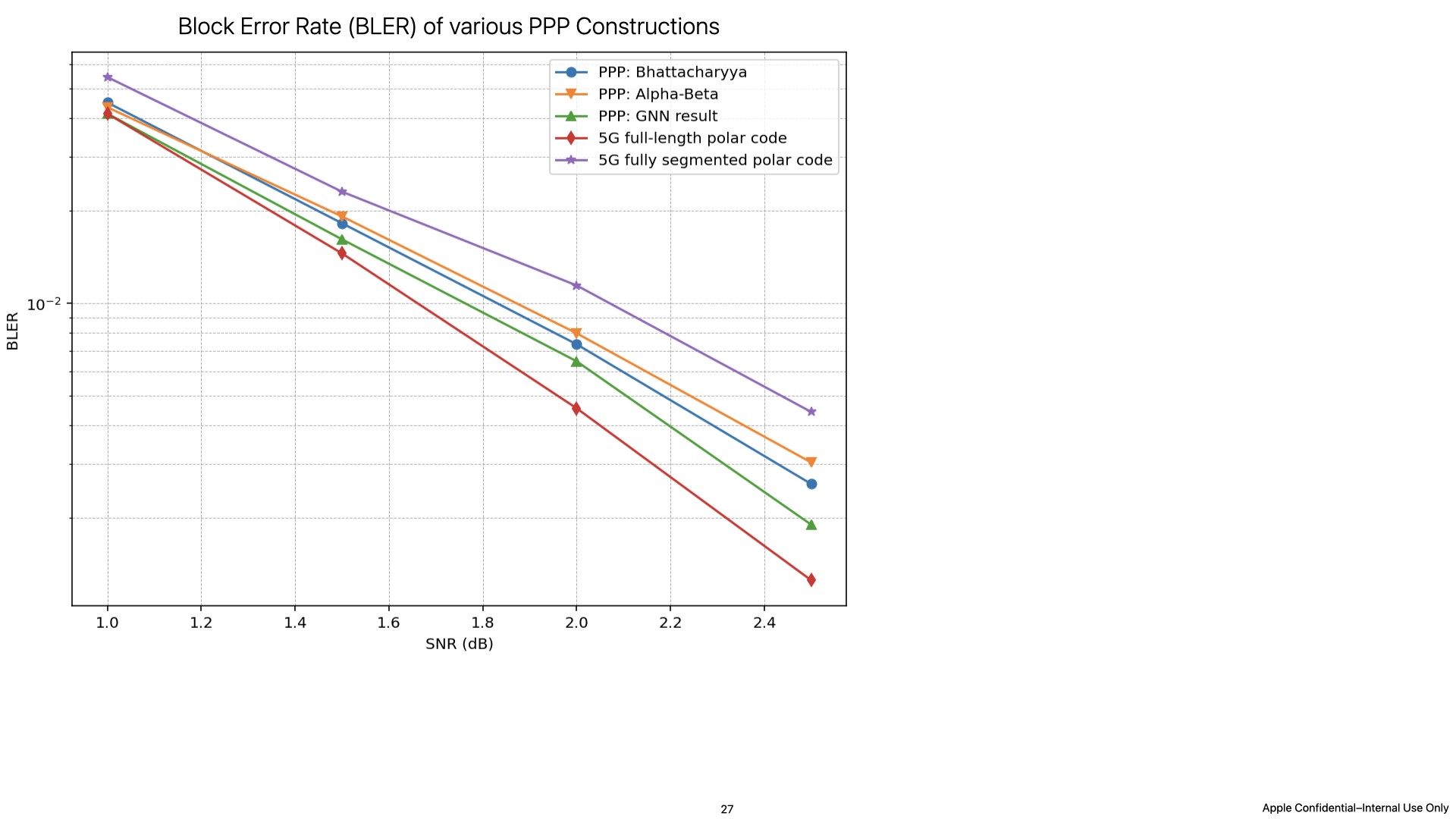}
    \caption{Performance comparison for $(N,K,L)=(64,32,16)$. 
    PPP significantly outperforms segmentation and approaches the performance of the larger polar code.}
    \label{fig:pppCon}\vspace{-1.5ex}
\end{figure}

\section{Numerical Results}
\label{sec:sim}
\noindent 
In this section, we present further numerical evidence of the superior performance of PPP codes compared to legacy solutions. Figure~\ref{fig:pppBLER} illustrates the smooth transition from fully segmented polar codewords, without any inter-segment coding, to a double-length polar code achieved through partial polarization. The simulation parameters are $N=1024$, $K=160$, and a 24-bit CRC. As the partial polarization ratio $\tau$ increases from 0 to 1, the BLER performance progressively converges to that of the longer polar code. At the same 
time, the effective code rate of the leading segment decreases. In this design example, we stop at the largest value of $\tau$ that still accommodates the 16 RNTI bits in the leading segment, which corresponds to $\tau = 7/8$.

Next, recall that legacy aggregation methods in polar coding are typically based on repetition schemes. While repetition improves the effective SNR, it does not provide additional coding gain. In contrast, PPP codes recover much of the coding gain that is lost under segmentation, without requiring extended hardware. This advantage is particularly valuable in blind decoding, since decoding of the second segment is often unnecessary. Figure~\ref{fig:pppAgg} highlights this effect: under the same configuration as before, a PPP code outperforms legacy repetition-based aggregation schemes, such as those currently employed in the 5G NR PDCCH.

\begin{figure}[t!]
    \centering
    \includegraphics[width=\columnwidth,trim={6.0cm 0.0cm 5.0cm 0.0cm},clip]{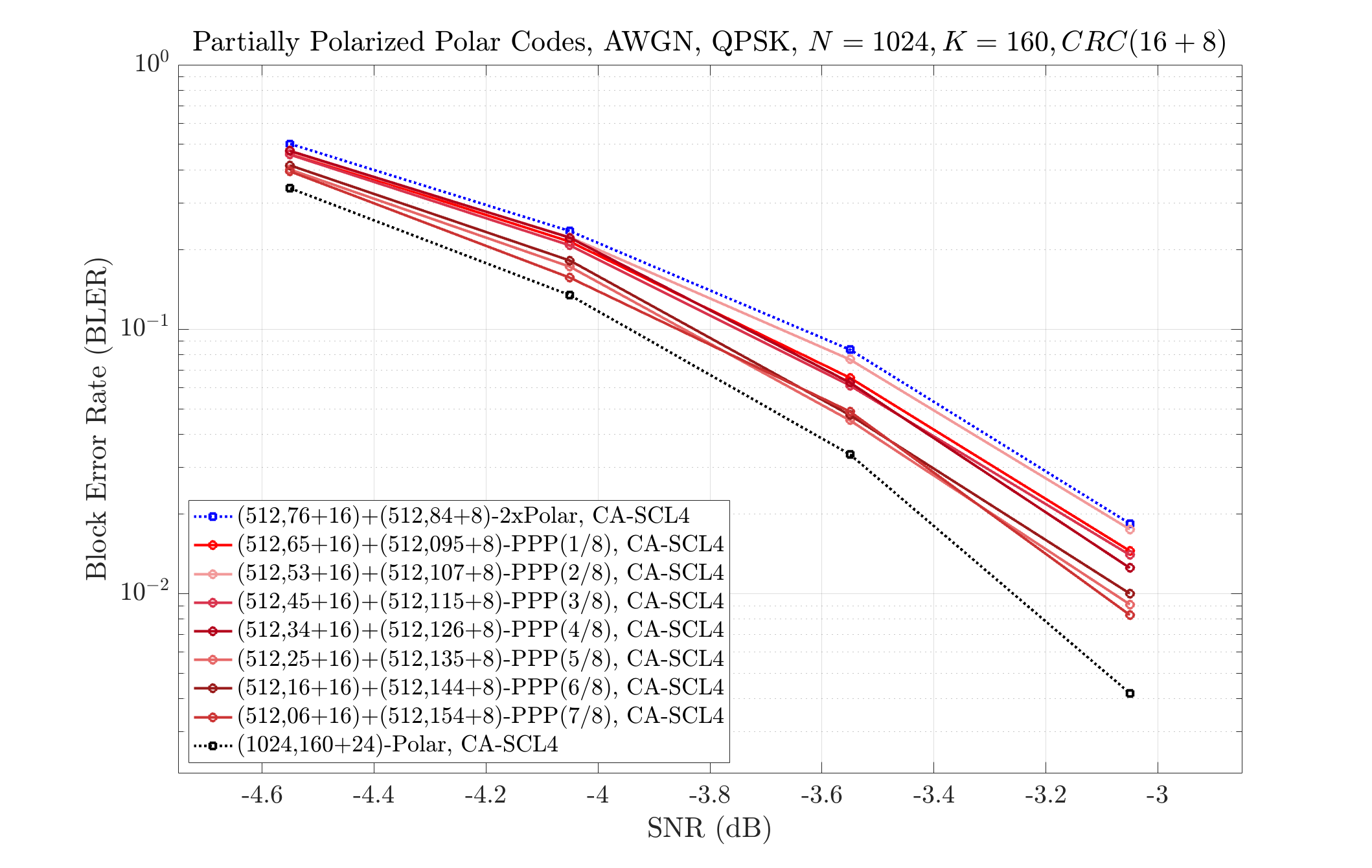}
    \caption{Performance of PPP codes for various values of $\tau$.}
    \label{fig:pppBLER}
\end{figure}

\begin{figure}[t]
    \centering
    \includegraphics[width=\columnwidth,trim={5.0cm 0.0cm 4.0cm 0.0cm},clip]{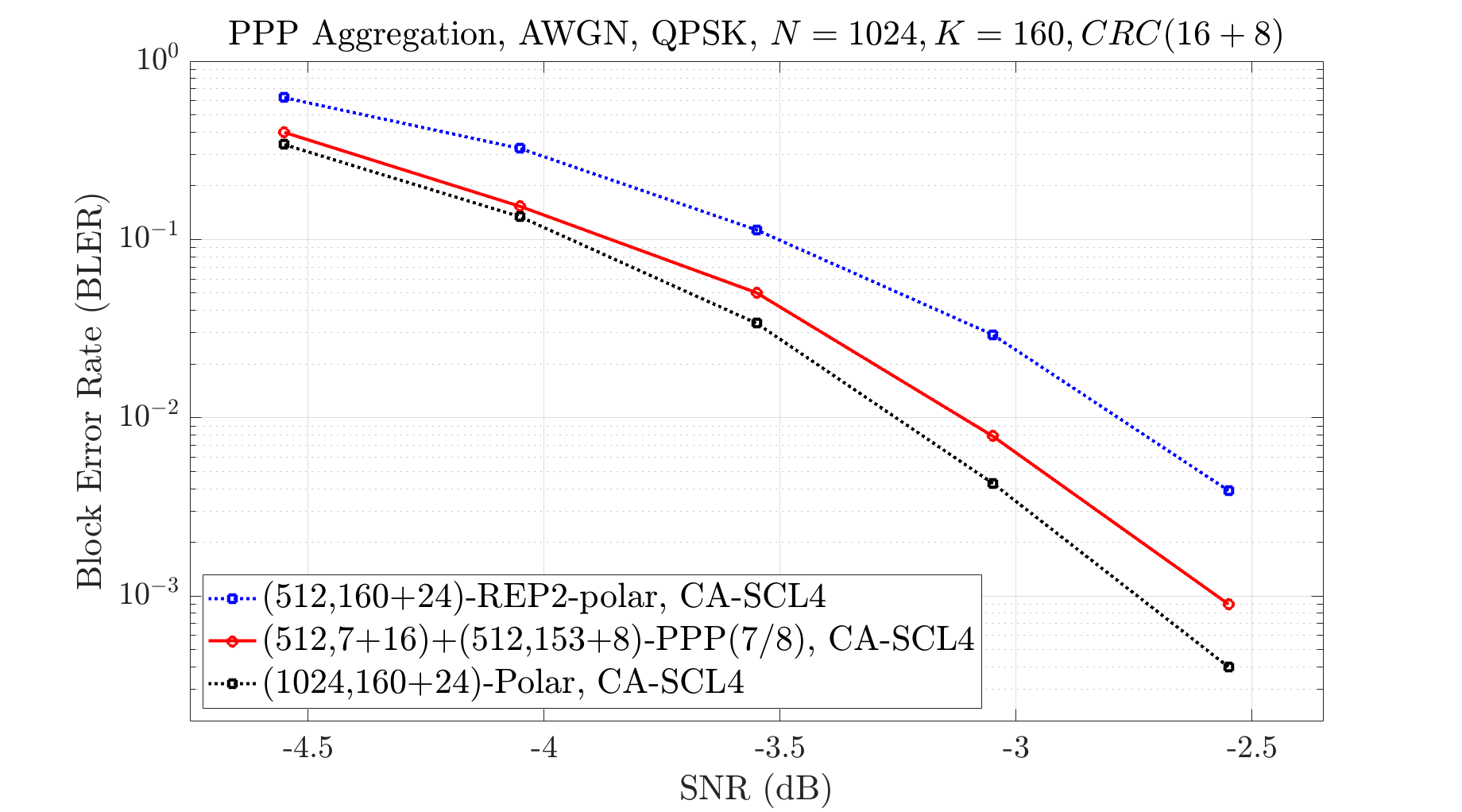}
    \caption{PPP codes outperform repetition-based aggregation schemes used in 5G NR PDCCH.}
        \label{fig:pppAgg}
\end{figure}
\section{Concluding Remarks}
\label{sec:conclusion}
\noindent
In this paper, we introduced partial polarization for inter-segment decoding as a means to enhance performance and enable early termination in two-stage decoding, making it a strong candidate for next-generation PDCCH blind decoding. Building on legacy polar codes, we proposed PPP codes, proved their capacity-achieving property, and presented multiple construction algorithms. Numerical results show that PPP codes significantly outperform both segmentation and aggregation, the two prevailing methods for handling larger payloads without added hardware complexity. In a way, PPP codes retain much of the coding gain that would otherwise be lost 
to legacy solutions. Furthermore, PPP codes and multi-stage decoding effectively decouple blind detection from DCI decoding, enabling the use of advanced algorithms such as dynamic SC-flip decoding for secondary segments.


\bibliographystyle{IEEEtran}
\bibliography{refs}

\end{document}